\documentclass[a4paper]{article}
\usepackage{times}  
\usepackage{helvet}  
\usepackage{courier}  
\usepackage{url}  
\usepackage{graphicx}  
\title{MUDA: A Truthful Multi-Unit Double-Auction Mechanism}
\author{
Erel Segal-Halevi, Ariel University, Ariel, Israel\\
Avinatan Hassidim, Bar-Ilan University, Ramat-Gan, Israel\\
Yonatan Aumann, Bar-Ilan University, Ramat-Gan, Israel\\
}

\usepackage[sort&compress]{natbib} 
\usepackage{caption}

\usepackage{amsthm}
\usepackage{amsmath}
\usepackage{amssymb}

\usepackage{graphicx}
\usepackage[shortcuts]{extdash}
\usepackage{framed}

\usepackage{version}
\includeversion{tedious-calculations}

\theoremstyle{definition}
\newtheorem{theorem}{Theorem}
\newtheorem{example}{Example}

\newtheorem*{remark*}{Remark}

\theoremstyle{plain}

\newtheorem*{lemma*}{Lemma}

\newcommand{\abs}[1]{\left|#1\right|}
\newcommand\range[2]{\in\{#1,\dots,#2\}}
\newcommand{\e}[1]{\exp{\left(#1\right)}}

\begin{document}
\maketitle
\begin{abstract}
In a seminal paper, McAfee (1992) presented a truthful mechanism for double auctions, attaining asymptotically-optimal gain-from-trade without any prior information on the valuations of the traders. McAfee's mechanism handles single-parametric agents, allowing each seller to sell a single unit and each buyer to buy a single unit.
This paper presents a double-auction mechanism that handles multi-parametric agents and allows multiple units per trader,
as long as the valuation functions of all traders have decreasing marginal returns. The mechanism
is prior-free, ex-post individually-rational, dominant-strategy truthful and strongly-budget-balanced. Its gain-from-trade approaches the optimum when the market size is sufficiently large.
\end{abstract}

\section{Introduction}
\label{sec:intro}
In a two-sided market, there are several sellers who hold items for sale and several buyers who consider buying these items. Examples are stock exchanges, used-car markets,
emission trading markets \citep{godby1999market,sturm2008double},
 Internet advertisements \citep{feldman2016double} and markets for spectrum reallocation \citep{leyton2017economics}.
Each trader has a different valuation to each bundle of items. 
In contrast to a one-sided market, here the valuations of both the buyers and the sellers are their private information, and both sides might act strategically.
A \emph{double auction} is a mechanism for organizing a two-sided market --- deciding who will buy, who will sell and at what prices.

An important requirement from a double-auction is \emph{efficiency}, which is measured by its \emph{gain-from-trade} (GFT) --- the total value gained by the buyers minus the total value contributed by the sellers. As an example, in a used-car market with a single buyer and a single seller holding a single car, if the seller values the car as $v_s$ and the buyer as $v_b>v_s$, then the potential GFT is $v_b-v_s$. 

The most commonly used double-auction mechanism is the \emph{Walrasian mechanism} \citep{rustichini1994convergence,babaioff2014efficiency}. It computes an \emph{equilibrium price} --- a price at which 
the supply equals the demand: the total number of units that sellers are interested to sell at this price equals the total number of units that buyers are interested to buy at this price. In a single-good market, an equilibrium price exists whenever the agents have \emph{decreasing-marginal-returns} (DMR) --- the marginal utility for an agent from having one more unit is weakly-decreasing in his current number of units \citep{gul2000english}. 
Moreover, by the First Welfare Theorem, this mechanism attains the maximum GFT \citep[Theorem 11.13]{nisan2007introduction}. Unfortunately, this mechanism is not \emph{incentive-compatible}  (IC) --- agents have an incentive to misreport their valuations in order to manipulate the price. 

The problem of designing an IC double-auction has already been considered by Vickrey. In his seminal paper \citep{vickrey1961counter} he described a variant of his famous second-price auction for a two-sided market. Like its one-sided variant, it is dominant-strategy incentive-compatible (DSIC, aka \emph{truthful}) --- it is a weakly-dominant strategy for each agent to reveal its true valuation function. But unlike the one-sided variant it is not \emph{budget-balanced} (BB) --- it has a deficit --- it requires the market-maker to bring money from home. While it is possible to charge entrance-fees to cover this cost, this makes the mechanism not \emph{individually rational} (IR) --- some traders might lose from participating. 

\citet{myerson1983efficient} proved that, in a two-sided market, \emph{any} mechanism that is IR, BB and IC cannot be efficient . Intuitively, the reason is that it is impossible to truthfully determine prices for trading.
Consider any mechanism that charges the buyer $p_b$ and pays the seller $p_s$. If $p_s < v_b$, then the seller is incentivized to bid $(p_s + v_b)/2$ to force the price up
(the mechanism has to do the deal since it is still efficient). Similarly, if $p_b > v_s$, the buyer is incentivized to force the price down. Setting $p_s=v_b$ and $p_b=v_s$ leads to a deficit. One way out is to determine take-it-or-leave-it prices independently of the traders' valuations, but this might result in a total loss of GFT.

This impossibility result initiated a search for double-auction mechanisms that are IC, IR and BB, and attain an \emph{approximately} maximal GFT. 
We define the \emph{competitive ratio} of a mechanism as the minimum ratio (over all utility profiles) of its GFT divided by the optimal GFT. 
The first approximation mechanism was presented by McAfee
\citet{mcafee1992dominant} for the case when each seller has a single unit and each buyer wants a single unit. 
McAfee's mechanism is truthful --- it is a weakly-dominant strategy for each agent to report his true valuation for the item. 
Its competitive-ratio is $1-1/k$, where $k$ is the number of units traded in the optimal situation.  (i.e, its GFT is always at least $1-1/k$ of the maximum GFT). Thus, McAfee's mechanism is \emph{asymptotically efficient} --- when the market-size $k$ grows to infinity, the GFT approaches the maximum.
In addition to being truthful, IR, BB and asymptotically-efficient, McAfee's mechanism has a fifth nice property --- it is \emph{prior-free} (PF). This means that it does not require any statistical information about the traders' valuations; it works well even for worst-case (adversarial) valuations. 

The main drawback of McAfee's mechanism is that it works only for single-unit traders.  Another potential drawback is that it is only \emph{weakly-budget-balanced} (WBB) --- the market-maker need not bring money, but may the have to take money. Moreover, in some cases the market-maker consumes almost all the GFT, leaving only little GFT to the traders \citep{segal2016sbba}. 
This may be acceptable when the market-maker is a monopolist (e.g. a government),
but might be problematic 
when the market-maker faces competition, e.g. in stock exchange platforms, since a low GFT for the traders might drive them away to the competitors.
\footnote{Technically it is possible to convert a WBB mechanism to a SBB one: randomly pick an agent before the trade, disallow him to participate in the trade, and give him all the surplus after the trade. However, this might induce a lot of people without any interest in the trade (e.g, ``sellers'' with no goods or ``buyers'' with no money) to participate in the auction, in the hope of winning the lottery. We believe this goes against the idea of truthfulness, and might have adverse effects on the market. 
}

This paper presents MUDA --- a double-auction mechanism for traders that buy and sell multiple units. The only assumption required is that all traders have decreasing-marginal-returns; this is the same assumption that guarantees the existence of market-equilibrium. 

The main idea of MUDA is to calculate equilibrium prices using \emph{random halving}. The general scheme is presented below; it is explained in more detail in Section \ref{sec:MUDA-details}.
\begin{framed}
\noindent\textbf{MUDA} (general scheme):
Split the market to two sub-markets, left and right, by sending each trader to each side with probability 1/2, independently of the others.
Then, in each sub-market:
\begin{enumerate}
\item Calculate a market-equilibrium price ($p^R$ at the right, $p^L$ at the left).
\item Let the agents trade at the price from the other market ($p^L$ at the right, $p^R$ at the left).
\end{enumerate}
\end{framed}
The random-sampling technique was found useful in  one-sided markets \citep{Goldberg2001Competitive,Goldberg2006Competitive,devanur2015envy,balcan2008reducing,balcan2007random} 
and matching markets \citep{devanur2009adwords}. 
However, applying it to a two-sided market poses a new challenge: \emph{material balance} --- the number of units bought and sold must be exactly equal. This is in contrast to a one-sided market, where the seller may leave some units unsold. Random-sampling in two-sided markets was used only with one-parametric agents --- agents whose valuations are characterized by a single signal \citet{baliga2003market,kojima2014double}.

With random-sampling, the price at each sub-market is determined exogenously based on the other sub-market, so the agents cannot manipulate it by reporting strategically. However, because of the material-balance requirement, the traders have another reason to report strategically --- they might try to manipulate the \emph{quantity} of trade. Thus, the multi-unit two-sided setting is more difficult than both the single-unit two-sided setting of McAfee and the multi-unit one-sided setting common in the random-sampling literature. To illustrate this added difficulty, we prove:
\begin{theorem}
\label{thm:negative}
Suppose a single seller holds $M$ units of a good and a single buyer considers buying at most $M$ units of this good and the price per unit of the good is determined exogenously. Then, the expected competitive ratio of every DSIC and IR mechanism, whether deterministic or randomized, is:
\begin{itemize}
\item at most $1/M$ for agents with general valuations, and ---
\item at most $1/H_M$ for agents with DMR valuations, where $H_M$ is the $M$-th harmonic number ($H_M \approx \ln{M}+\frac{1}{M}$).
\end{itemize}
Both these upper bounds are tight.
\end{theorem}
The proof is given in \citet{segalhalevi2017truthful}.
Theorem \ref{thm:negative} can be seen as a dual to the Myerson--Satterthwaite impossibility theorem. In their setting, the quantity is determined exogenously and the traders might manipulate the price; in our setting, the price is determined exogenously and the traders might manipulate the quantity.

We present two variants of MUDA that overcome this impossibility when the market is sufficiently large. We call them \emph{Lottery-MUDA} and \emph{Vickrey-MUDA}. 

Step 1 is the same in both variants. Step 2 
in both variants 
starts by calculating the aggregate demand and aggregate supply in each sub-market, at the prices calculated in the other sub-market. 
Usually, the aggregate demand will be larger or smaller than the aggregate supply, so the sub-market will have a \emph{long side} and a \emph{short side} (e.g, if there is more demand than supply then the buyers are the long side and the sellers are the short side). The short side can always trade their optimum quantity; i.e, if the buyers are short, we can let each buyer buy the number of units that maximizes his utility given the price. The two variants of MUDA differ in the way they handle the long side: in \emph{Lottery-MUDA} the traders in the long side are selected using a random permutation that ignores their values, while in \emph{Vickrey-MUDA} the high-value traders in the long side are selected. 
In Vickrey-MUDA, each selected trader pays the market-maker a positive trading-fee
calculated as in a Vickrey auction; this fee is separate from the money transferred among the traders
(the inspiration to this idea came from a recent working paper by \citet{loertscher2016dominant}).

\begin{theorem}
\label{thm:strategic}
Both variants of MUDA are prior-free, dominant-strategy incentive-compatible, individually-rational and budget-balanced (no deficit). In addition, Lottery-MUDA is strongly-budget-balanced (no surplus).
\end{theorem}
\begin{proof}
Prior-freeness holds by design: MUDA does not use any statistical information on the valuations. 

DSIC and IR will be proved in Section \ref{sec:MUDA-properties} after presenting the mechanism details. 

Lottery-MUDA is strongly-budget-balanced (SBB) since all money goes from buyers to sellers; the market-maker neither brings nor takes any money. 
Vickrey-MUDA is only \emph{weakly budget-balanced (WBB)} --- the market-maker does not lose money, but may make profit from  trading-fees. 
\end{proof}
We emphasize that the properties are true \emph{ex-post} --- for every outcome of the randomization in the mechanism.

We distinguish between the \emph{total-GFT}, which includes 
the gain of both the agents and the market-maker,
and the \emph{agents-GFT}, which 
includes only the gain of the agents. 
Since Lottery-MUDA is SBB, the agents-GFT equals the total-GFT.
In Vickrey-MUDA, the total-GFT is always higher since the most profitable deals are selected. However, the agents-GFT might be much lower. In a particular example shown in the arXiv version, the total-GFT is high but the agents-GFT is near 0. 
Thus, Vickrey-MUDA may be preferred when the market-maker is a monopolist (e.g, a government), while Lottery-MUDA may be preferred when the market-maker faces competition (e.g, a stock-exchange).

The competitive-ratio of MUDA depends on
$k$ --- the number of units traded in the optimal situation (the $k$ used by McAfee 1992),
and
$M$ --- the maximum number of units offered(demanded) by a single seller(buyer)%
.
$M$ represents the market concentration --- how much market power is held by a single trader. 
\begin{theorem}
\label{thm:competitive}
The expected total-GFT of both Lottery-MUDA and Vickrey-MUDA is at least a fraction $1-O(M\sqrt{\frac{\ln k}{k}})$ of the maximum total-GFT.
\end{theorem}
The proof is given in Section \ref{sec:competitive-ratio}. 

While Vickrey-MUDA attains more total-GFT than Lottery-MUDA, their asymptotic behavior is similar --- both of them approach the optimal total-GFT when the market size ($k$) is large, as long 
as the market-concentration factor (represented by $M$) is kept constant. In contrast, if $M$ is very large ($M=k$) the market effectively has a single buyer and a single seller, and the impossibility result of \citep{myerson1983efficient} implies that no positive approximation of the GFT is possible.
The arXiv version shows an example in which the \emph{agents-GFT} of Vickrey-MUDA is close to 0 (the agents-GFT of Lottery-MUDA always equals its total-GFT).

While we focus on approximating the maximum GFT (buyers' values minus trading sellers' values),
other mechanisms in the literature approximate the maximum \emph{social welfare} (buyers' values plus non-trading sellers' values). 
But \emph{any} mechanism that attains a fraction $\alpha$ of the optimal GFT also attains a fraction of at least $\alpha$ of the optimal social welfare \citep{brustle2017approximating}. Hence, MUDA is asymptotically-optimal with respect to the social-welfare too. 

The lower bound of Theorem \ref{thm:competitive}
depends on $k$ --- the optimal trade size. 
Ideally, we would like a bound that depends on the number of traders that come to the market (say, $n$). However, we cannot attain such bound theoretically in a worst-case analysis, even for the social-welfare. As an example, consider a single-good single-unit market having $n$ sellers with value $2$, $n-1$ buyers with value $1$ and one buyer with value $1000000$. Here $k=1$ as there is only one relevant trade. The competitive ratio of any mechanism depends only on the probability with which it performs this single trade; this probability does not change even when $n\to\infty$.

We complement our worst-case analysis that depends on $k$ with simulations of both variants of MUDA 
on agents drawn from both synthetic and realistic distributions. The simulations show that, when valuations are random (and not worst-case), the competitive-ratio of MUDA increases with the number of traders. These are presented in Section \ref{sec:simulations}.

\subsection{Related Work} \label{sec:related}
Most existing mechanisms for multi-unit double-auction are not truthful, e.g. \citet{plott1990multiple}.

The research on truthful double auction has made many advancements since \citet{mcafee1992dominant}. There are variants of McAfee's mechanism for  maximizing the auctioneer's surplus \citep{Deshmukh2002Truthful}, handling spatially-distributed markets  \citep{Babaioff2004Mechanisms}, transaction costs \citep{Chu2006Agent}, supply chains \citep{Babaioff2006Incentive}, constraints on the set of traders that can trade simultaneously \citep{Yao2011Efficient,dutting2017modularity,brustle2017approximating}, 
online arrival of buyers \citep{Wang2010TODA} and strong budget-balance \citep{ColiniBaldeschi2016Approximately}.
All these advancements are for single-unit agents.

Other double-auction mechanisms either assume that the agents' valuations are additive
\citep{huang2002design,xu2010secondary,feldman2016double,goel2016reservation,hirai2017polyhedral}
or assume that their valuations are represented by a single parameter 
\citet{Gonen2007Generalized,kojima2014double}.
The DMR valuations handled by MUDA are multi-parametric and include additive valuations as a special case.

Several recent double-auction mechanisms work in a Bayesian setting --- they assume that the traders' valuations are drawn from some probability distribution that is public knowledge. Such knowledge allows the mechanism designer to attain approximate efficiency without relying on the agents' reports.
Examples are 
\citet{yoon2008participatory,colini2017approximately}.
Some other mechanisms require partial prior knowledge on the valuations, such as their median \citep{Blumrosen2014Reallocation} or their maximum and minimum value \citet{gonen2017dycom}.
\citep{baliga2003market} assume that the agents' valuations are drawn from some \emph{unknown} distribution (they handle single-unit traders) .
In contrast, MUDA needs no prior information on the agents' valuations and does not even assume that they are drawn from some distribution.

We are aware of two truthful mechanisms that handle multi-parametric agents with DMR valuations in a prior-free way. The first is by
\citet{Blumrosen2014Reallocation}: their competitive ratio is 1/48 --- it is not asymptotically-efficient. The second is by  \citet{loertscher2016dominant}, which is being developed simultaneously to our work.
Their mechanism (that does not use market-halving) is asymptotically-efficient when the valuations of the traders
are drawn from probability distributions satisfying certain conditions. In contrast, our convergence theorem does not assume that agents' valuations are drawn from a distribution at all.

\section{Model}
\label{sec:notation}

\paragraph{\textbf{Agents and valuations}}\label{sub:valuations}
We consider a market for a single good. Some agents, the ``sellers'', are endowed with at most $M$ units of that good, and other agents, the ``buyers'', are endowed with an unlimited supply of money.
Each agent $j$ has a valuation-function $v_j$ that returns, for every integer $t$ between $0$ and $M$, the agent's value for owning $t$ units. The valuations are normalized such that $v_i(0)=0$. 

All agents have 
\emph{decreasing marginal returns} (DMR). Formally, 
$v_j(t+1)-v_j(t) \leq v_j(t)-v_j(t-1)$, for every agent $j$ and $t\range{1}{M-1}$.

We will often represent a multi-unit agent as $M$ single-unit \emph{virtual agents}; the value of virtual-agent $t$ of agent $j$ is the agent's marginal value for having the $t$-th unit:
$v_{j,t} := v_j(t) - v_j(t-1)$ for $t\range{1}{M}$ (a similar idea was used by \citet{chawla2010multiparameter}).

We assume that the agents' valuations are \emph{generic}, i.e, all marginal values of different traders are different and linearly-independent over the integers (no linear combination with integer coefficients equals zero). 
This assumption can be dropped if we use centralized tie-breaking; see \citet{Hsu2016Do} for other ways to handle ties in markets.


Given a price $p$ per unit of the good, the \emph{gain} of a buyer $i$ from buying $t$ units is $Gain_i(t,p)= v_i(t) - t\cdot p$, and the gain of a seller $j$ from selling $t$ units is 
$Gain_j(t,p) = t\cdot p + v_j(M_j-t) - v_j(M_j)$, where $M_j$ is the number of units initially held by seller $j$, $M_j\leq M$.

A \emph{mechanism} is a (randomized) function that takes the agents' valuations and returns (1) a trading-price $p$, (2) for each buyer(seller) $i$, the amount of units $t_i$ he should buy(sell) at price $p$, and possibly a trading-fee $f_i$ paid to the market-maker. A mechanism is \emph{materially-balanced} if the number of units bought equals the number of units sold:
$\sum_{i\in\text{Buyers}}t_i = \sum_{j\in\text{Sellers}}t_j$.
It is \emph{budget-balanced (BB)} if 
it is materially-balanced and $\sum_{i\in\text{Traders}} f_i\geq 0$. It is \emph{strongly-budget-balanced (SBB)} if  $\sum_{i\in\text{Traders}} f_i= 0$.

A mechanism is
\emph{individually-rational (IR)} if every agent $i$ has a weakly-positive gain: $Gain_i(t_i,p)-f_i\geq 0$ with probability 1.  It is \emph{DSIC} (\emph{= truthful}) 
if an agent can never increase his gain by pretending to have different valuations.

The \emph{demand} of buyer $i$ is the number of units that maximizes the gain: $Demand_i(p) = \arg\max_{t\range{0}{M}} Gain_i(t,p)$. The genericity assumption implies that the maximum is unique. 
When $i$ has DMR, $Demand_i(p)$ equals the number of virtual-buyers $i,t$ with $v_{i,t}>p$.
\footnote{
This is true only for a DMR agent. 
For example, suppose buyer $i$ values one unit as 3 and two units as 4. Then $v_{i,1}=3$ and $v_{i,2}=1$. 
If the price is 2, then there is one virtual buyer with value above the price, and indeed the agent's demand is 1. However, if the buyer values one unit as 1 and two units as 4, then $v_{i,1}=1$ and $v_{i,2}=3$. 
If the price is 2, then there is still one virtual buyer with value above the price, but the agent's demand is 0.
}

The \emph{aggregate-demand} at price $p$ is the sum of demands of all buyers. Equivalently, this is the number of virtual-buyers with $v_{j,t}> p$.  
The \emph{supply} and \emph{aggregate-supply} are defined analogously.

The \emph{total-gain-from-trade} is the sum of gains of all buyers and sellers: $total~GFT := \sum_{i\in\text{Traders}} Gain_i(t_i)$. 
Note that in a materially-balanced mechanism, the total-GFT does not depend on the trading-price.
The \emph{agents-GFT} is the total GFT minus the total fees, $agents~GFT := total~GFT -  \sum_{i\in\text{Traders}} f_i$. 

\section{Mechanism Details}
\label{sec:MUDA-details}
This section explains the details of the MUDA mechanism presented in the introduction. 
The steps are done in each sub-market separately. For convenience we describe step 1 in the right market and step 2 in the left market; the execution in the opposite direction is entirely analogous.

\paragraph{Step 1: Price calculation.}
We calculate a price that is an \emph{equilibrium price} at the right market --- a price $p^R$ for which
$
Demand^R(p^R) = Supply^R(p^R) 
$.
Such a price exists even in more general (multi-good) settings. It can be found, for example, by simulating an English auction \citep{gul2000english}, or by binary search.

\paragraph{Step 2: Posted-price trade.}
For each buyer $i$ in the left market, calculate $Demand_i(p^R)$. 
For each seller $j$ in the left market, calculate 
$Supply_j(p^R)$. 
Let $Demand^L$ be the sum of demands and $Supply^L$ the sum of supplies. If $Demand^L=Supply^L$, then we can let the traders trade freely at price $p^R$ and the market will clear. 
Usually, however, we will not be so lucky: there will be either excess demand ($Demand^L>Supply^L$) or excess supply ($Demand^L<Supply^L$). These two cases are handled analogously; henceforth we describe how to handle excess supply. First, we ask each buyer to pay in advance for the optimal number of units he wants to buy at price $p^R$, so we have money for $Demand^L$ units. Since $Demand^L<Supply^L$, we will have more than enough units to give to all these buyers for their money. The problem is how to select the sellers that will supply these $Demand^L$ units. We present two solutions.

(a) Lottery: Order the sellers randomly; let each seller in turn sell at price $p^R$ as many units as he likes, while there is money (i.e, while at most $Demand^L$ units are sold). 

(b) Vickrey-style auction: Order the \emph{virtual} sellers in increasing order of their value. From the $Supply^L$ virtual sellers whose value is below $p^R$, pick the $Demand^L$ virtual-sellers with the lowest values, and have each of them sell an item at price $p^R$.

The Vickrey-style auction is followed by a 4th step: \textbf{trading-fees}. For each seller $j$, let $k_j$ be the number of virtual-sellers of this seller that are picked. Note that $k_j \leq Supply_j(p^R)$.  The fee paid by seller $j$ is determined by the potential gains of the virtual-sellers that are ``pushed'' out of the market because of seller $j$. Specifically, consider the set of virtual-sellers who want to trade in the left-market, \emph{except} the $k_j$ virtual-sellers of $j$. From this set, pick the (at most) $Demand^L$ low-value ones. In this set, there are (at most) $k_j$ virtual-sellers that do not trade when seller $j$ is present. Seller $j$
pays to the market-maker, the gain (price minus value) of these virtual sellers. Note that the trading-fee is positive if $k_j > 0$ and zero if $k_j=0$.

\begin{example}
\label{exm:vickrey-MUDA}
Suppose $p^R=50$. The virtual-buyers in the left market have values $100,90,80,60,40,20$, so the aggregate demand is 4. There are two sellers: Alice's values are $10,20,40,60,70$ and Bob's values are $15,25,35,45,65$, so the aggregate supply is 7 and there is excess supply. The four high-value virtual-buyers (with values 100,90,80,60) each pays 50 and is guaranteed a unit.

In Lottery-MUDA, the two sellers are ordered randomly; if Alice is first then she sells 3 units and gains $40+30+10=80$, and Bob sells 1 unit and gains $35$, so the GFT is 115. If Bob is first then he sells 4 units and gains $35+25+15+5=80$, and Alice sells nothing. The price is 50 per unit, all money collected from the buyers is given to the sellers, and no money is given to the market-maker.

In Vickrey-MUDA, the four low-value virtual-sellers are picked: they are the virtual-sellers with values 10,15,20,25. So each each real seller sells two units for 50 per unit. At this point Alice's gain is $(50-10)+(50-20)=70$ and Bob's gain is $(50-15)+(50-25)=60$, so the total-GFT is 130.
Now, Alice pays a fee of 20, since her presence prevents Bob from selling two units worth for him 35 and 45, so her net gain is 50. Bob pays 10, since his presence prevents Alice from selling a unit worth for her 40 (the unit worth 60 would not have been sold anyway), so his net gain is 50 too, and the agents-GFT is 100. The market-maker gains 30.
\qed
\end{example}

\section{Strategic Properties of MUDA}
\label{sec:MUDA-properties}
In both variants of MUDA, the traders cannot affect their trading-price.
Moreover, in both variants, the traders in the short side of each sub-market trade the amount of units that maximizes their gain given the price, so for them the mechanism is obviously IR and DSIC. As for the long-side traders:

(a) in Lottery-MUDA they play random serial dictatorship: the first agents in the line trade their optimal quantity given the price and the last agents in the line cannot trade at all. There is at most a single agent, in the middle of the line, who trades less than his optimal quantity. Because of the DMR assumption, it is always optimal to trade as many units as possible up to the optimal quantity (since the highest gain comes from trading the first units).
Therefore, the mechanism is IR and DSIC for all traders.

(b) In Vickrey-MUDA, long-side virtual-traders trade only if they have positive gain. In this case, a trader with $k_j$ participating virtual-traders pays a fee that is equal to the gain of 
at most $k_j$ non-participating virtual-traders. Since the
mechanism always selects the virtual-traders with the highest gain, the total fee paid by any trader is lower than his gain, so the net gain remains positive and the mechanism is IR. The mechanism is DSIC since it is effectively a multi-unit Vickrey-auction with a reserve-price. It is known that such a mechanism is DSIC; we omit the proof.

\section{Competitive-Ratio Analysis}
\label{sec:competitive-ratio}
In this section we prove Theorem \ref{thm:competitive}.
In fact, we prove a more general 
claim that depends on a third parameter, $m$ --- the minimum number of units offered(demanded) by a single seller(buyer) at same value. 
We assume that the virtual-traders of each trader $i$ are divided to groups of size at least $m$, such that the values in each group are the same.
So each agent $i$ values $m_{i,1}$ units as $v_{i,1}$, $m_{i,2}$ units as $v_{i,2}$, etc., with
$m_{i,l} \geq m$ for all $l$ and $\sum_{l} m_{i,l} \leq M$.
Theorem \ref{thm:competitive} follows from the expression at the end of this section by setting $m=1$.

We first analyze the optimal trade, then the right sub-market and finally the left sub-market.

\subsection{Optimal trade}
In the optimal trade, there is a set $B_*$ of virtual-buyers
who buy goods from a set $S_*$ of virtual-sellers.
By material-balance the numbers of virtual-agents in 
both groups are the same; this is the number we denoted by $k$:
\begin{align}
\label{eq:clearing-opt}
|B_*| = |S_*| = k
\end{align}
We call these $k$ buyers and sellers the \emph{efficient traders}. We make the pessimistic assumption that all GFT in the sub-markets comes from these efficient traders. Therefore, the GFT of our mechanism depends on the numbers of efficient traders that trade in each sub-market.

The reduction in GFT has two reasons: one is the \emph{sampling error} --- efficient buyers and sellers land in different sub-markets, so they do not meet and cannot trade. This error is easy to bound using standard tail bounds. The second reason is the \emph{pricing error} --- the price at the sub-market might be too high or too low, which might create imbalance in the demand and supply. Analyzing this error requires careful analysis of the equilibrium in the optimal situation vs. the equilibrium in each sub-market.

\subsection{Right sub-market}
In the right sub-market, MUDA calculates an equilibrium price $p^R$. We define four sets of virtual-traders:
\begin{itemize}
\item $B_-$ is the set of efficient virtual-buyers (members of $B_*$) whose value is below $p^R$ (so they won't buy at price $p^R$).
\item $S_-$ is the set of efficient virtual-sellers (members of $S_*$) whose value is above $p^R$ (so they won't sell at price $p^R$).

\item $B_+$ is the set of inefficient virtual-buyers whose value is above $p^R$ (so they want to buy at price $p^R$).
\item $S_+$ is the set of inefficient virtual-sellers whose value is below $p^R$ (so they want to sell at price $p^R$).
\end{itemize}
These sets represent the pricing error, so we want to  upper-bound their sizes.

For any set $T$ of agents, denote by $T^R$ the subset of $T$ that is sampled to the right market and by $T^L$ the subset of $T$ sampled to the left market. By definition of the equilibrium price $p^R$:
\begin{align}
\label{eq:clearing-right}
|B^R_*|-|B^R_-|+|B^R_+|
=
|S^R_*|-|S^R_-|+|S^R_+|
\end{align}
In order to relate (\ref{eq:clearing-opt}) and (\ref{eq:clearing-right}), we have to relate $B^R_*,S^R_*$ to $B_*,S_*$. This is be done using the following lemma.
\begin{lemma*}
\label{lem:det}
For every set $T$ of virtual-traders and for every integer $q>0$:
\begin{align}
\label{eq:sampling-det}
\text{w.p.~~} 1-2 \e{\frac{-2 m q^2}{M^2 |T|}}:
&&
\abs{|T^R|-{\frac{|T|}{2}}} < q
\end{align}
(``w.p. $x$'' is a shorthand to ``with probability at least $x$'').
\end{lemma*}
The lemma is proved using Hoeffding's inequality. The proof is standard and we omit it.

We apply this lemma twice, to $B_*$ and $S_*$, and combine the outcomes using the union bound. This gives, $\forall q>0$:
\begin{align}
\label{eq:sampling-*}
\text{w.p.} 1-4 e^{{-2 m q^2\over M^2 k}}:
&&
\abs{|B_*^R|-{k\over 2}} < q
\text{~and~}
\abs{|S_*^R|-{k\over 2}} < q
\end{align}
Combining equations \ref{eq:clearing-opt},\ref{eq:clearing-right},\ref{eq:sampling-*} gives, $\forall q>0$:
\begin{align}
\label{eq:sampling-R}
\text{w.p.} 1-4 e^{{-2 m q^2\over M^2 k}}:
&&
\abs{
|B_-^R\cup S_+^R|
-
|B_+^R \cup S_-^R|
} < 2 q
\end{align}
Of the two sets in the left-hand side, at least one must be empty: if $p^R$ is too high (relative to some optimal equilibrium price $p^O$) then efficient buyers quit and inefficient sellers join, but no inefficient buyers join and no efficient sellers quit, so $B_+^R = S_-^R=\emptyset$. This situation is illustrated in Figure \ref{fig:1d}. Analogously, if $p^R$ is too low then $B_-^R= S_+^R=\emptyset$.

\begin{figure}
\begin{center}
	\includegraphics[width=4cm]{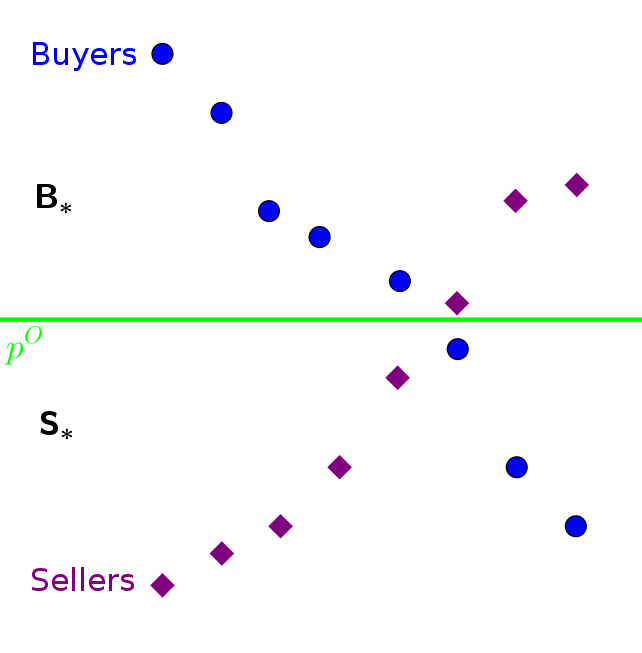}
	\hskip 2mm	\includegraphics[width=4cm]{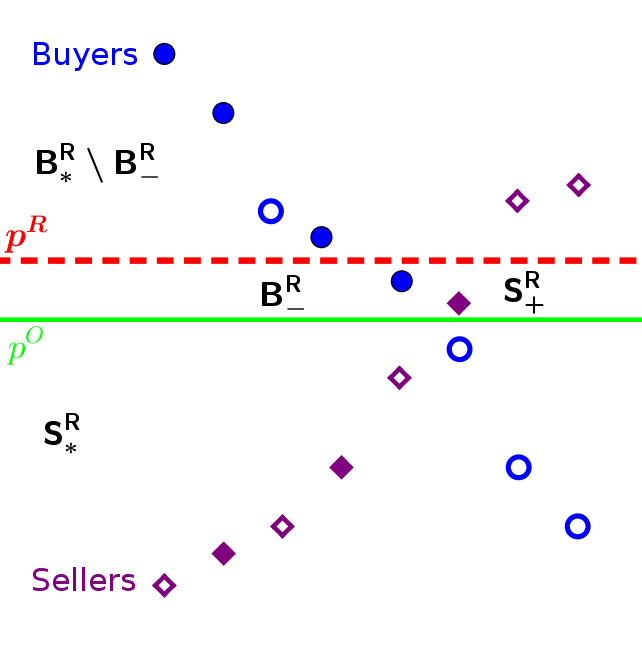}
\end{center}
\caption{
\label{fig:1d}
Price-distortion due to sampling. Each ball is the value of a single-unit buyer and each square is the value of a single-unit seller.
\textbf{Left:} the optimal situation where  there are 5 deals: the buyers above the line $p^O$ trade with the sellers below the line.
\textbf{Right:} the right market: full markers represent the traders sampled to the right and empty markers represent the traders sampled to the left. The equilibrium price $p^R$ is higher than $p^O$, so some efficient buyers quit ($B_-$) and some inefficient sellers join ($S_+$).
}
\end{figure}

From now on we assume that the situation is like in Figure \ref{fig:1d} (the other situation is analogous). So (\ref{eq:sampling-R}) implies:
\begin{align}
\label{eq:b-r}
\forall q:~~
\text{w.p.} 1-4 e^{{-2 m q^2\over M^2 k}}:~~
|B_-^R| < 2 q \text{~and~} |S_+^R|  < 2 q
\end{align}
Our goal is now to derive an upper bound on $B_-$ and $S_+$. They are entirely analogous; we focus on $B_-$. Note that we cannot apply (\ref{eq:sampling-det}) directly to $B_-$, since $B_-$ is a random-set --- it depends on the random-sampling through $p^R$. 
(\ref{eq:sampling-det}) does not apply to sets $T$ that depend on the random-sampling; as an illustration, suppose the set $T$ is selected such that it contains only virtual-agents from the right market. Then $T^R=T$ so (\ref{eq:sampling-det})  is obviously not satisfied. Fortunately, $B_-$ is a special random-set --- it is \emph{one-dimensional}: for every integer $t>0$, it has only  a single possible value with cardinality $t$, which is the set of $t$ virtual-buyers with the lowest value in $B_*$ (see Figure \ref{fig:1d}, where $t=1$). Denote these sets by $B_{-,t}$. For every $t$, $B_{-,t}$ is independent of the random-sampling, so it is eligible to (\ref{eq:sampling-det}). Substituting there  $q\to 2 t-2 q$ gives:
\begin{align}
\label{eq:b-rt}
\forall q,t:
\text{w.p.} 1-2 e^{{-2 m (2t - 2q)^2\over M^2\cdot 4 t}}:
&&
\abs{|B_{-,4 t}^R|-{2 t}} < 2 t - 2 q
\\
\notag
\implies &&
|B_{-,4 t}^R| >  2 q
\end{align}
If $|B_-^R|< 2 q$ and $|B_{-,4t}^R| >  2 q$, then necessarily $|B_-| < 4 t$. So by combining (\ref{eq:b-r}) and (\ref{eq:b-rt}) with the union bound we get, $\forall q,t$:
\begin{align}
\label{eq:b-qt}
\text{w.p.} 1
-4 e^{{-2 m (t - q)^2\over M^2 t}}
-4 e^{{-2 m q^2\over M^2 k}}
:
|B_-| < 4 t \text{~\&~} |S_+|  < 4 t
\end{align}
To simplify the expression we choose $t=2 q$; assuming $ 2 q<k$, this implies ${-2 m (t - q)^2\over M^2 t} < {-2 m q^2\over M^2 k}$, so (\ref{eq:b-qt}) simplifies to:
\begin{align}
\label{eq:b-q}
\forall q<{k/2}:~
\text{w.p.} 1
-8 e^{{-2 m q^2\over M^2 k}}:
~
|B_-| < 8 q \text{~\&~} |S_+|  < 8 q
\end{align}

\subsection{Left sub-market}
Denote by $B^L$($S^L$) the set of efficient buyers(sellers) who want to buy(sell) in the left sub-market at price $p^R$:
\begin{align*}
B^L &:= B_*^L \setminus B_-^L = B_* \setminus B_*^R \setminus B_-^L
\\
S^L &:= S_*^L \setminus S_-^L = S_*^L = S_* \setminus S_*^R 
\end{align*}
(Recall that we assume the case in which $S_-$ and $B_+$ are empty; the other case is analogous).

Using (\ref{eq:sampling-*}) and (\ref{eq:b-q}) with the same $q$ gives, for every $q<{k\over 2}$: 
\begin{align*}
\text{w.p. $1-8 e^{{-2 m q^2\over M^2 k}}$:}~~
|B^L| > {k\over 2}  - 9 q
\text{~~and~~}
|S^L| > {k\over 2}  - q
\end{align*}
In Vickrey-MUDA, the most efficient traders in each side trade with each other. Therefore, the mechanism makes at least the ${k\over 2}-9 q$ most efficient deals in the left submarket. Similar considerations are true in the right submarket. All in all, Vickrey-MUDA does at least the $k-18 q$ most efficient out of the $k$ efficient deals. Therefore, w.p. $1-8 e^{{-2 m q^2\over M^2 k}}$, the competitive ratio is at least $1 - 18 {q \over  k}$.

In Lottery-MUDA, the efficient sellers in $S_*^L$ have to compete with the inefficient sellers in $S_+^L$ in the random lottery. The expected number of efficient deals carried out is thus:
\begin{align*}
{|B^L|\over |S_+^L|+|S_*^L|} \geq {k/2 - 9 q \over k/2 + 9 q}  \geq {k\over 2} - 18 q
\end{align*}
Therefore, w.p. $1-8 e^{{-2 m q^2\over M^2 k}}$, 
the \emph{expected} competitive ratio is at least $1 - 36 {q \over  k}$. So w.p. 1, the expected competitive ratio is at least 
$1 -8 e^{{-2 m q^2\over M^2 k}} - 36 {q \over  k}$.

All our analysis so far holds simultaneously for every $q<{k\over 2}$. Now, we select $q$ to maximize the expected competitive ratio. With some tedious calculations we find that, when $k$ is sufficiently large, we can select $q$ such that the competitive ratio is at least:
\begin{tedious-calculations}
\footnote{
Define $x = {q\over k}$.
We look for an $x$ that minimizes $8 e^{{-2 m x^2 k\over M^2}} + 36 x$. The first-order condition is:
\begin{align*}
{36 M^2 \over 32 m k} = x\cdot  \e{- 2 m k x^2 \over M^2}
\end{align*}
The solution is:
\begin{align*}
x_{\min} = \sqrt{{-M^2\over 4 m k} \cdot W\left({- 36^2 M^2 \over 16^2 m k}\right)}
\end{align*}
When $k$ is sufficiently large, the argument of $W$ is small and it is approximated by:
\begin{align*}
W(z) \approx - \ln({1\over -x}) - \ln(\ln({1\over -x})) \approx 
-\ln({1\over -x})
\end{align*}
so:
\begin{align*}
x_{\min} \approx \sqrt{{M^2\over 4 m k} \cdot \ln \left({16^2 m k\over 36^2 M^2}\right)}
\end{align*}
Substituting $q=k x_{\min}$ gives:
\begin{align*}
36 {q\over k} = {36 M\over 2\sqrt{m k}} \sqrt{\ln \left({16^2 m k\over 36^2 M^2}\right)}
\\
8 e^{{-2 m q^2\over M^2 k}} = 
8 \left({16^2 m k\over 36^2 M^2}\right)^{-1/2}
=
{36 M\over 2 \sqrt{m k}}
\end{align*}
}
\end{tedious-calculations}
\begin{align*}
&1 - {36 M\over 2\sqrt {m k}}\left(1 + \sqrt{\ln\left(16^2 m k\over 36^2 M^2\right)}\right)
\\
=&
1 - O\left(M \sqrt{{\ln {m k} \over m k}}\right)
\end{align*}
The analysis for Vickrey-MUDA is obtained by replacing 36 with 18; the asymptotic behavior is the same.

\section{Simulations}
\label{sec:simulations}
To complement our theoretic analysis, we simulated MUDA on traders with valuations sampled both from a synthetic distribution and an empirical distribution based on real stock-exchange data.



\subsection{Uniform distribution}
In the first experiment, 
for each agent, we sampled 
$M/m$ values from a uniform distribution with support $[V-A, V+A]$.
We considered each of these values as the marginal-value of $m$ virtual-traders, so that each agent has $M$ virtual-traders. We ordered the values in decreasing order to get DMR valuations.
In the experiments, we took $m=100$ and $V=500$ and varied the noise-amplitude $A$ between $50$ and $450$. Here we show the results for $A=250$; varying $A$ did not have much effect on the results. We repeated each experiment 100 times and averaged the results.

In the first sub-experiment, we kept $M$ constant (at $100,000$) and varied the number of real traders between 0 and 1000.
The results are shown in Figure \ref{fig:results-uniform}/Left. 
The competitive ratios of all variants of MUDA
increase towards $1$ when the number of traders grows.

In the second sub-experiment, we kept constant the total number of units held by all sellers.  We increased $M$ from 100 to $10^8$, and decreased the number of traders accordingly so that the total number of units remains constant. The results are shown in Figure \ref{fig:results-uniform}/Right. The competitive ratios decrease when $M$ increases and the same number of units are concentrated in the hand of fewer traders. When $M$ is very large we effectively have one buyer and one seller, and then the impossibility result of \citep{myerson1983efficient} implies that no positive competitive ratio is possible.

\begin{figure}
\begin{center}
\includegraphics[width=\columnwidth]{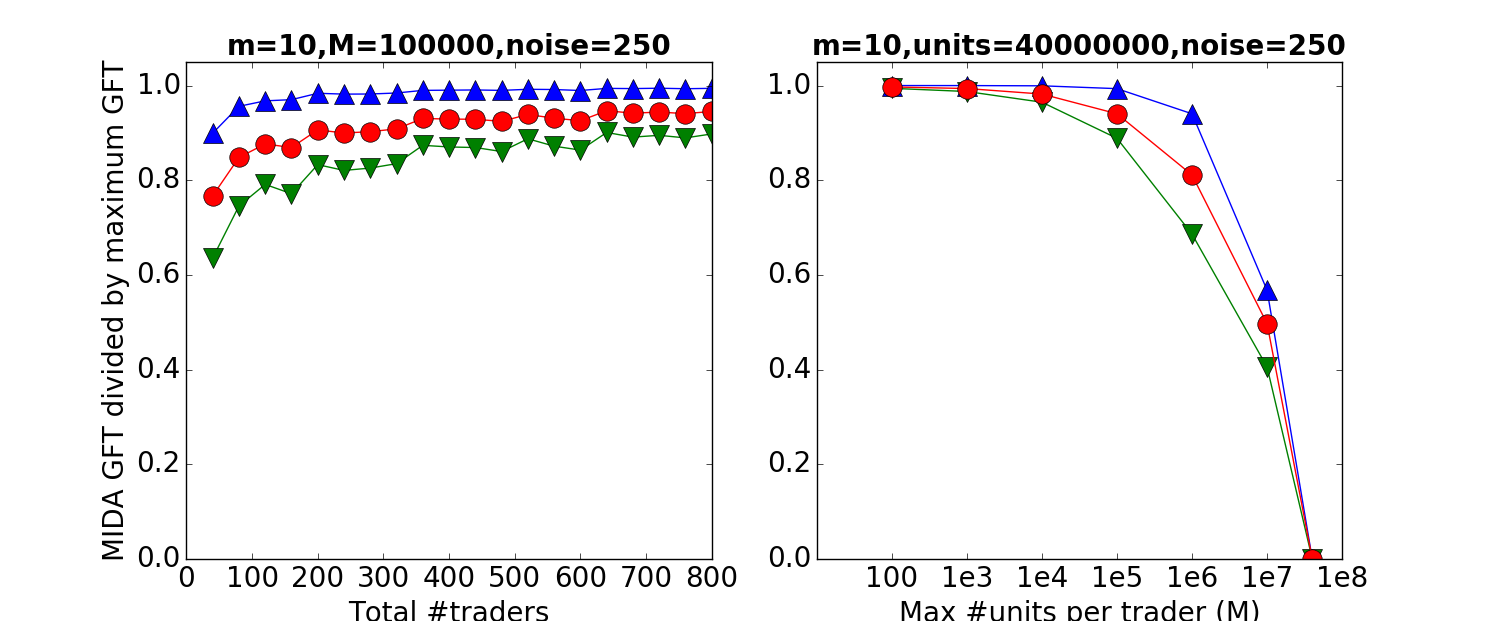}
\end{center}
\caption{
\label{fig:results-uniform}
Competitive ratio of MUDA when traders' valuations are drawn from a uniform distribution. Legend:\\
* Downward-wedges = agents-GFT of Vickrey-MUDA;\\
* Discs = Lottery-MUDA;\\
* Upward-wedges = total-GFT of Vickrey-MUDA.
\\
\textbf{Left}: maximum number of units per trader ($M$) is fixed, and number of traders (and total number of units) increases.\\
* Agents-GFT of Vickrey-MUDA ranges from 0.58 to 0.86.  \\
* GFT of Lottery-MUDA ranges from 0.74 to 0.92. \\
* Total-GFT of Vickrey-MUDA ranges from 0.897 to 0.991. 
\\
\textbf{Right}: total number of units is fixed, and maximum number of units per trader increases (so the total number of traders decreases).
}
\end{figure}

Both plots show that the GFT of Lottery-MUDA is approximately middle-way between the total-GFT and the agents-GFT of Vickrey-MUDA. 
This fact has a practical implication regarding the choice of mechanism: if the market is managed by a monopolist (e.g. the government), then Vickrey-MUDA is better since its total-GFT is higher; but if the market competes with other markets (e.g. stock trading platforms), then Lottery-MUDA is better since its agents-GFT is higher.

We found that, when the number of units per trader is fixed, the performance is determined by the number of traders in the market. Therefore, the plot at the left of Figure \ref{fig:results-uniform} remains the same even when we set 
e.g. $m=1$ and $M=10$. As for the plot at the right of Figure \ref{fig:results-uniform}, if we set the total number of units to e.g. $1000$, then the plot approaches zero already when $M=1e3$, since in this case there is a single buyer and a single seller.

\subsection{Empirical stock-exchange distribution}

\begin{figure}
\begin{center}
\includegraphics[width=\columnwidth]{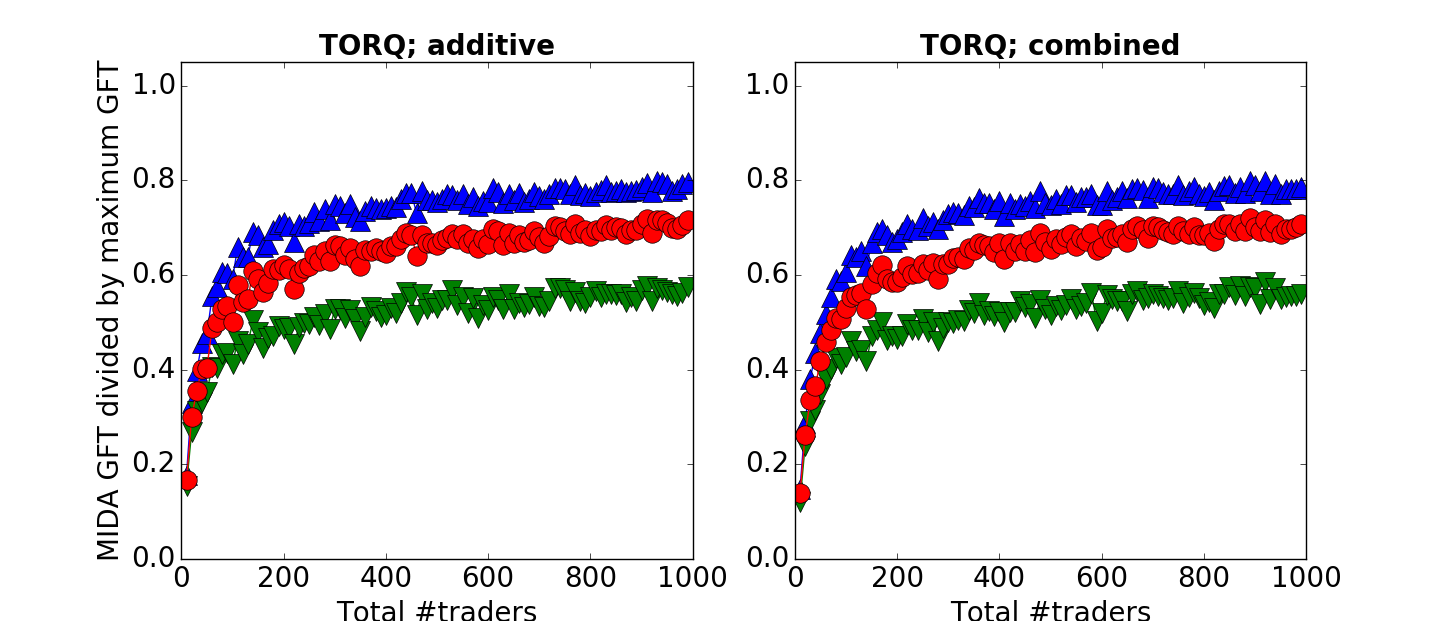}
\end{center}
\caption{
\label{fig:results-torq}
Competitive ratio of MUDA when traders' valuations are drawn from an empirical distribution based on NYSE orders.
Discs = Lottery-MUDA;
upward/downward-wedges = total/agents-GFT of Vickrey-MUDA.
At the left, each order is assumed to belong to a different trader (so all traders are additive).
At the right, orders from the same order-date are merged to a single trader.
}
\end{figure}
In the second experiment, we used the TORQ database
\citep{hasbrouck1992using,lee2000inferring}. It contains buy and sell orders given for a sample of 144 NYSE stocks for the three months 1990-11 through 1991-01. In the NYSE, each day before the continuous trade begins, there is a phase of ``start-of-day auctions''. For each stock, a separate multi-unit double-auction is conducted in the following way. All buy and sell orders given for that stock before the start of day are collected. An equilibrium price is calculated. All buy-orders above the price and all sell-orders below the price are executed. As explained in the introduction, this mechanism is efficient but not truthful. Therefore, the orders may not represent the true values of the traders.
While there are econometric methods for estimating the true values from the reported values, these are beyond the scope of the present paper, since we do not need to know the true value of every trader --- all we need is a distribution of values. 
Our results are valid as long as 
the empirical distribution of reported values resembles 
the distribution of values in the real world.

In the present paper we assume that the empirical distribution of the orders is similar to the empirical distribution of the true values.

For the experiment, we considered only the start-of-day orders. These orders are given in the format (Symbol, Date, Order date, Side, Price, Quantity) where 
\begin{itemize}
\item 
Symbol is the three-letter acronym of the stock;
\item
Date is the day of the auction.
\item
Order date is the day the order was made, which may be earlier than Date;
\item
Side is BUY or SELL; and 
\item
Price and Quantity are the actual bid.
\end{itemize}

The dataset contains 144 symbols, 7914 different symbol-date combinations and $\approx 210000$ orders.

The dataset does not contain the identities of the traders.
We made two experiments: in the first we treated each order as a separate trader, so that all traders are additive (each trader values a certain quantity of units for the same price-per-unit).  In the second experiment, we simulated traders with non-additive valuations by merging bids based on the order-date. I.e, we  heuristically assumed that two bids with the same symbol, date, order-date and side belong to the same trader. We ordered the bids of the same trader in descending order to create DMR valuations, similarly to the uniform experiment.
In average, each merged bidder contained 10 different orders.
The number of units per order ranges between 100 and 99000, so $m=100$. The number of units per trader, after combining orders from the same order-date, ranges between 5000 and 12000000 with median 148000, so $M=12000000$.

For each symbol, we created a collection of all traders in all days to represent the empirical distribution of this symbol.
Then, in each experiment we sampled $n$ traders, where $n = 10,20,\ldots,990$. All in all we simulated $99\cdot 144$ auctions and averaged the $144$ auctions for each $n$.
The results are shown in Figure \ref{fig:results-torq}.

\subsection{Discussion}
It is interesting that the plots in Figure \ref{fig:results-torq} look almost the same, i.e, the non-additivity had little effect on the results. The reason might be that stock-traders have many different bids with very similar prices, so that they are almost additive. 

The competitive ratio in the TORQ experiment is somewhat lower than in the uniform case. We attribute this to the large variation in the number of units: while some traders hold only 5000 units, others hold as many as 12 million.

The source code used for the experiments is available at the following URL:

\url{https://github.com/erelsgl/economics}

\section{Future Work}
\label{sec:future}
It is interesting whether MUDA can be extended to agents with general valuations (not DMR). This poses two challenges. First, in step 1, a price-equilibrium might not exist (the existence of a price-equilibrium is guaranteed only when the agents have DMR). Second, in step 2, material balance might fail. For example, if there is one buyer who only wants to buy 2 units and one seller who only wants to sell 3 units, then with DMR we can assume that both agree to trade 2 units, but without DMR this assumption might be false --- the seller might have negative utility from selling 3 units.

Another challenge is handling double auctions with multiple types of items. Preliminary results in this direction were presented by \citet{segal2016mida}.

We will also be happy to compare MUDA to other double-auction mechanisms. In particular, an  interesting practical question is 
when would MUDA's truthfulness actually matter --- how does its GFT compare to the GFT in Nash equilibrium of the standard (non-truthful) market mechanism? As far as we know, the Nash-equilibrium GFT of the non-truthful market mechanism (i.e, its ``price-of-anarchy'') is still an open question.    The closest reference we found is by
\citet{babaioff2014efficiency}, who consider the price-of-anarchy of the market mechanism in a single-sided market.  We are not aware of analogous results for double-sided markets.

But non-truthful market mechanisms have disadvantages beyond the price-of-anarchy.  
First, the players need to spend resources to obtain information about other players and calculate their best response. 
Second, the players may not even play an equilibrium. It is possible that players do not know the correct state of the market and hence play best-response to a fictitious state. This can bring social welfare to zero. Consider any mechanism that sets a single price clearing the market. For every active buyer(seller), it is a best response to bid below(above) the market price in order to push the price down(up).  If all agents play these best responses, the outcome is that nothing gets sold. Such market failures are prevented by the truthfulness of MUDA.

\section{Acknowledgments}
Assaf Romm helped us access and analyze the TORQ dataset and provided many helpful comments.
Ron Adin, Simcha Haber, Ron Peretz, Tom van der Zanden, Jack D'Aurizio, Andre Nicolas,  Brian Tung, Robert Israel, Clement C. and C. Rose helped us with the probability calculations.
The participants of the game-theory seminar in Bar-Ilan University, computational economics and economic theory seminars in the Hebrew university of Jerusalem, algorithms seminar in Tel-Aviv university and computer science seminar in Ariel university provided helpful comments on various aspects of the paper.
The anonymous reviewers in Theoretical Economics, SODA 2016, EC 2016, SODA 2017, EC 2017, and, of course, AAAI 2018, gave us detailed and helpful feedback.

Erel was supported by ISF grant 1083/13, Doctoral Fellowships of Excellence Program and Mordecai and Monique Katz Graduate Fellowship Program at Bar-Ilan University. 
Avinatan Hassidim is supported by ISF grant 1394/16.

\fontsize{9.0pt}{10.0pt} \selectfont
\bibliographystyle{apalike}
\bibliography{erel}
\end{document}